\documentclass[english]{tlp}
\usepackage{import}

\usepackage{ifthen}
\usepackage{url}
\usepackage{amssymb}
\usepackage{amsmath}
\usepackage{import}
\usepackage{xparse}
\usepackage{amsmath}
\usepackage[usenames,dvipsnames]{xcolor}
\usepackage{xspace}
\usepackage{datatool}
\usepackage{glossaries}

\newcommand{\m}[1]{\ensuremath{#1}\xspace}
\newcommand{\trval}[1]{\m{{\bf #1}}}

	\newcommand{\limplies}{\m{\Rightarrow}}
	
	\newcommand{\lequiv}{\m{\Leftrightarrow}}

	\newcommand{\lrule}{\m{\leftarrow}}
	\newcommand{\cause}{\m{\stackrel{c}{\lrule}}}
	\newcommand{\rul}{\m{\leftarrow}}
	\newcommand{\ltrue}{\trval{t}}
	\newcommand{\lfalse}{\trval{f}}
	\newcommand{\lunkn}{\trval{u}}
	
	\newcommand{\ra}{\m{\rightarrow}}
	
	\newcommand{\mim}{\limplies}
	\newcommand{\equi}{\lequiv}
	\newcommand{\Tr}{\ltrue}
	\newcommand{\Fa}{\lfalse}
	\newcommand{\Un}{\lunkn}

	\newcommand{\voc}{\m{\Sigma}}

	\newcommand{\struct}{\m{I}}
	
	\newcommand{\I}{\m{\mathcal{I}}}

	\newcommand{\D}{\m{\Delta}}

	\NewDocumentCommand\inter{g+g}{%
	  \IfNoValueTF{#1}
	    {\struct}
	    {\m{#1^{#2}}}}

	\newcommand{\defined}[1]{\ensuremath{\mbox{\it Def}({#1})}\xspace}
	
	\newcommand{\pars}[1]{\ensuremath{\mbox{\it Par}({#1})}\xspace}

	\newcommand{\xxx}{\m{\overline{x}}}

	\newcommand{\ddd}{\m{\overline{d}}}

	\newcommand{\ttt}{\m{\overline{t}}}

	\newcommand{\bool}{\m{\mathbb{B}}}

	\renewcommand{\int}{\m{\mathbb{Z}}}

	\newcommand{\leqp}{{\m{\,\leq_p\,}}}

	\DeclareMathOperator\glb{glb}

	\NewDocumentCommand\subs{g+g}{%
	  \IfNoValueTF{#1}
	    {\m{/}}
	    {\m{#1/ #2}}}

	\newcommand{\logicname}[1]{\text{\sc #1}\xspace}

	\newcommand{\fodot}{\logicname{FO(\ensuremath{\cdot})}}
	
	\newcommand{\foid}{\logicname{FO(\ensuremath{ID})}}
	
	\newcommand{\foids}{\logicname{FO(\ensuremath{ID^{*}})}}
	\newcommand{\soids}{\logicname{SO(\ensuremath{ID^{*}})}}
	\newcommand{\esoids}{\logicname{ESO(\ensuremath{ID^{*}})}}
	\newcommand{\asoids}{\logicname{ASO(\ensuremath{ID^{*})}}}

\newcommand{\ouracronym}[3]{%
	\newacronym{#1}{#2}{#3}
	\expandafter\newcommand\csname #1\endcsname{\gls{#1}\xspace}%
}
	\ouracronym{FO}{FO}{first-order logic}
	\ouracronym{MX}{MX}{Model Expansion}
	\ouracronym{MO}{MO}{Model Optimization}
	\ouracronym{ASP}{ASP}{Answer Set Programming}
	\ouracronym{TP}{TP}{Theorem Proving}
	\ouracronym{LP}{LP}{Logic Programming}
	\ouracronym{CP}{CP}{Constraint Programming}
	\ouracronym{FP}{FP}{Functional Programming}
	\ouracronym{KR}{KR}{Knowledge Representation}
	\ouracronym{CSP}{CSP}{Constraint Satisfaction Problem}
	\ouracronym{SMT}{SMT}{SAT Modulo Theories}
	\ouracronym{KBS}{KBS}{Knowledge Base System}
	\ouracronym{NNF}{NNF}{Negation Normal Form}
	\ouracronym{FNNF}{FNNF}{Flat Negation Normal Form}
	\ouracronym{DefNNF}{DefNNF}{Definition Negation Normal Form}
	\ouracronym{CDCL}{CDCL}{Conflict-Driven Clause-Learning}
	\ouracronym{LCG}{LCG}{Lazy Clause Generation}

	\makeatletter
	\def\ifenv#1{
	\def\@tempa{#1}%
	\def\@ttempa{#1*}%
	\ifx\@tempa\@currenvir
	\expandafter\@firstoftwo
	\else
	\expandafter\@secondoftwo
	\fi
	}
	\makeatother

	\newcommand{\ddrule}[4]{\ensuremath{#1 \leftarrow #2 & \{#3\} & #4}}
	\newcommand{\drule}[2]{\ensuremath{#1 & \leftarrow & #2}}

	\newcommand{\darule}[4]{\ensuremath{#1 \leftarrow #2 & \{#3\} & #4}}
	\newcommand{\arule}[2]{\ensuremath{#1 \, &\leftarrow \, #2}}

	\newcommand{\LNDRule}[2]{
	\ifenv{array}
	{\drule{#1}{#2}}
	{ \ifenv{align}
		{\arule{#1}{#2}}
		{\ifenv{align*}
		{\arule{#1}{#2}}
		{ERROR: using LDRule in unsupported environment: \@currenvir}
		}
	}
	}

	\newcommand{\LDRule}[4]{
	\ifenv{array}
	{\ddrule{#1}{#2}{#3}{#4}}
	{ \ifenv{align}
		{\darule{#1}{#2}{#3}{#4}}
		{\ifenv{align*}
		{\darule{#1}{#2}{#3}{#4}}
		{ERROR: using LDRule in unsupported environment: \@currenvir}
		}
	}
	}

	\NewDocumentCommand\LRule{m+g+g+g}{%
		\IfNoValueTF{#2}%
		{#1.&}{%
		\IfNoValueTF{#3}
		{\LNDRule{#1}{#2.}}
		{\LDRule{#1}{#2.}{#3}{#4}}%
		}
	}

	\NewDocumentCommand\CLRule{m+g}{%
	\ifenv{array}
	{\cdrule{#1}{#2}}
	{ \ifenv{align}
		{\carule{#1}{#2}}
		{\ifenv{align*}
			{\carule{#1}{#2}}
			{ERROR: using CLRule in unsupported environment: \@currenvir}
		}
	}
	}

	\NewDocumentCommand\carule{m+g}{%
		\IfNoValueTF{#2}
			{\ensuremath{#1.}}
			{\ensuremath{#1 \, &\cause \, #2}}}
	\NewDocumentCommand\cdrule{m+g}{%
		\IfNoValueTF{#2}
			{\ensuremath{#1.}}
			{\ensuremath{#1 & \cause & #2}}}

	\newcommand{\algrule}[4]{
	\hbox{{#1}:}& 
	\quad #2 ~\longrightarrow~ #3 
	\hbox{~ if } #4\\
	}

	\newcommand{\AlgoRule}[4]{
	\ifenv{array}
	{\algrule{#1}{#2}{#3}{#4}}
		{ERROR: using AlgoRule in unsupported environment: \@currenvir}
	}

\newcommand{\commentstyle}{\color{Gray}}

	\usepackage[final]{listings}

	\lstdefinelanguage{idp}{
		morekeywords=[1]{namespace,vocabulary,theory,structure,procedure,term,set,formula, spec, specification, template},
		morekeywords=[2]{include,using,type,isa,contains,partial,extern,LFD,GFD,constructed,from,constraint,func,pred,supertype,of,subtype,define},
		morekeywords=[3]{int,float,char,string,nat},
		morekeywords=[4]{if,then,else,for,end},
		morecomment=[s]{/*}{*/},	
		morecomment=[l]{//}
	}
	\newcommand{\ignore}[1]{}

	\newboolean{nocomments}
	\setboolean{nocomments}{false}

	\newboolean{commentmargin}
	\setboolean{commentmargin}{true}

	\newcommand{\namedcomment}[3]{
		\ifthenelse{\boolean{nocomments}}
		{} 
		{ 
			\ifthenelse{\boolean{commentmargin}}
				{ {\color{#3} \marginpar{\color{#3}\sc #2}#1}  } 
				{  {\color{#3} {\sc #2}: #1}  } 
		}
	}
	\newcommand{\mnamedcomment}[3]{\ifthenelse{\boolean{nocomments}}{}{{\marginpar{ \color{#3}{\sc #2}:#1}}}}

	\usepackage{soul}

\usepackage{etoolbox}

\newcommand\setcitation[2]{%
  \csdef{mycommoncitation#1}{#2}}

\setcitation{IDP}{WarrenBook/DeCatBBD14}
\setcitation{fodot}{tocl/DeneckerT08}
\setcitation{CP}{Apt03}
\setcitation{EZCSP}{lpnmr/Balduccini11}
\setcitation{KR}{Baral:2003}
\setcitation{ASPComp2}{lpnmr/DeneckerVBGT09}
\setcitation{ASPComp3}{journals/tplp/CalimeriIR14}
\setcitation{ASPComp4}{conf/lpnmr/AlvianoCCDDIKKOPPRRSSSWX13}
\setcitation{CPSupport}{ictai/DeCat13}
\setcitation{functionDetection}{iclp/DeCatB13}
\setcitation{fodot2asp}{DeneckerLTV12} 
\setcitation{Tarskian}{DeneckerLTV12} 
\setcitation{Inca}{iclp/DrescherW12}
\setcitation{csp2asp}{ijcai/DrescherW11}
\setcitation{DPLLT}{cav/GanzingerHNOT04}
\setcitation{AspInPractice}{synthesis/2012Gebser}
\setcitation{clasp}{ai/GebserKS12}
\setcitation{oclingo}{kr/GebserGKOSS12}
\setcitation{gringo}{lpnmr/GebserST07}
\setcitation{cmodels}{aaai/GiunchigliaLM04}
\setcitation{inputster}{tplp/Jansen13}
\setcitation{DLV}{tocl/LeonePFEGPS06}
\setcitation{LearningPaper}{TPLP/BruynoogheBBDDJLRDV} 
\setcitation{clog}{iclp/BogaertsVDV14} 
\setcitation{foc}{iclp/BogaertsVDV14}
\setcitation{inferenceClog}{ecai/BogaertsVDV14}
\setcitation{examplesClog}{nmr/BogaertsVDV14b} 
\setcitation{AFT}{DeneckerBV12}
\setcitation{KBS}{iclp/DeneckerV08}
\setcitation{KBPE}{inap/DePooterWD11}
\setcitation{lazyGrounding}{corr/CatDSB14} 
\setcitation{ASP}{marek99stable}
\setcitation{satid}{sat/MarienWDB08}
\setcitation{lazyclausegeneration}{constraints/OhrimenkoSC09}
\setcitation{FP}{ACMCS/Hudak89}
\setcitation{GroundingWithBounds}{jair/WittocxMD10}
\setcitation{SAT}{faia/SilvaLM09}
\setcitation{LTC}{iclp/Bogaerts14}
\setcitation{SPSAT}{ictai/DevriendtBMDD12}
\setcitation{BreakID}{cspsat/DevriendtBB14}
\setcitation{LCG}{stuckeyLCG}
\setcitation{MiniZinc}{conf/cp/NethercoteSBBDT07}
\setcitation{amadini}{cpaior/AmadiniGM13}
\setcitation{bootstrapping}{lash/BogaertsJDJBD14}
\setcitation{GroundedFixpoints}{aaai/BogaertsVD15}
\setcitation{LogicBlox}{datalog/GreenAK12}
\setcitation{proB}{journals/sttt/LeuschelB08}
\setcitation{NaturalInductions}{KR/DeneckerV14} 
\setcitation{LP}{jacm/EmdenK76}
\setcitation{SMT}{faia/BarrettSST09}
\setcitation{AF}{ai/Dung95}
\setcitation{ADF}{kr/BrewkaW10}
\setcitation{ADFRevisited}{ijcai/BrewkaSEWW13}
\setcitation{DefaultLogic}{ai/Reiter80}
\setcitation{DL}{ai/Reiter80}
\setcitation{AEL}{mo85}
\setcitation{}{}
\setcitation{completion}{adbt/Clark78}
\setcitation{}{}
\setcitation{}{}
\setcitation{}{}
\setcitation{}{}
\setcitation{}{}
\setcitation{}{}
\setcitation{}{}
\setcitation{}{}
\setcitation{}{}
\setcitation{}{}
\setcitation{}{}
\setcitation{}{}
\setcitation{}{}
\setcitation{}{}
\setcitation{}{}
\setcitation{}{}
\setcitation{}{}
\setcitation{}{}
\setcitation{}{}
\setcitation{}{}
\setcitation{}{}
\setcitation{}{}
\setcitation{}{}
\setcitation{}{}
\setcitation{}{}
\setcitation{}{}
\setcitation{}{}
\setcitation{}{}
\setcitation{}{}
\setcitation{}{}
\setcitation{}{}
\setcitation{}{}
\setcitation{}{}
\setcitation{}{}
\setcitation{}{}
\setcitation{}{}
\setcitation{}{}
\setcitation{}{}
\setcitation{}{}
\setcitation{}{}
\setcitation{}{}
\setcitation{}{}
\setcitation{}{}
\setcitation{}{}
\setcitation{}{}
\setcitation{}{}
\setcitation{}{}
\setcitation{}{}
  
\usepackage{amsthm}

\theoremstyle{plain}
\newtheorem{thm}{Theorem}[section]

\newtheorem*{lem*}{Lemma}

\newtheorem{proposition}[thm]{Proposition}

\theoremstyle{definition}
\newtheorem{defn}[thm]{Definition}
\newtheorem{definition}[thm]{Definition}

\newtheorem*{nota*}{Notation}

\newtheoremstyle{example_basic} 
{\topsep} 
{\topsep} 
{\normalfont}
{0pt}
{\bfseries}
{.}
{5pt plus 1pt minus 1pt}
{}

\newtheoremstyle{example_contd}
{\topsep} 
{\topsep} 
{\normalfont}
{0pt}
{\bfseries}
{.}
{5pt plus 1pt minus 1pt}
{\thmname{#1} \thmnumber{ #2}\thmnote{#3} (continued)}

\theoremstyle{example_basic} 

\newtheorem{example}[thm]{Example}

\newtheorem{ex*}{Example}
\theoremstyle{example_contd}

\theoremstyle{plain}

\usepackage[english]{babel}
\usepackage{caption}
\usepackage{subcaption}
\usepackage{bold-extra}
\usepackage{enumitem}
\DeclareCaptionFont{white}{\color{white}}
\DeclareCaptionFormat{listing}{\colorbox[cmyk]{0.43, 0.35, 0.35,0.01}{\parbox{\textwidth}{\hspace{15pt}#1#2#3}}}
\author[I. Dasseville, M. van der Hallen, G. Janssens, M. Denecker]{Ingmar Dasseville, Matthias van der Hallen, Gerda Janssens, Marc Denecker\\ 
KU Leuven\\
\email{firstname.lastname@cs.kuleuven.be}
}

\usepackage{listings}
\newcommand{\defp}[1]{\ensuremath{\mbox{\it Def}({#1})}\xspace}

\newcommand{\hilog}{HiLog}

\begin{document}
\setboolean{nocomments}{false}

\newcommand{\powSet}[1]{\m{\{s(\overline{s}) \leftarrow #1\}}}
\newcommand{\sRules}[1]{\m{\{(\overline{s})\leftarrow #1\}}}
\newcommand{\Log}{\m{\mathcal L}}
\newcommand{\figref}[1]{Figure~\ref{#1}}

\title{Semantics of templates in a compositional framework for building logics} 
\maketitle
\begin{abstract}
  There is a growing need for abstractions in logic specification
  languages such as \fodot and ASP.  One technique to achieve these
  abstractions are templates (sometimes called macros).  While the
  semantics of templates are virtually always described through a
  syntactical rewriting scheme, we present an alternative view on
  templates as second order definitions. To extend the existing
  definition construct of \fodot to second order, we introduce a
  powerful compositional framework for defining logics by modular
  integration of logic constructs specified as pairs of one
  syntactical and one semantical inductive rule.  We use the framework
  to build a logic of nested second order definitions suitable to
  express templates. We show that under suitable restrictions, the
  view of templates as macros is semantically correct and that adding
  them does not extend the descriptive complexity of the base logic,
  which is in line with results of existing approaches.
\end{abstract}

\section{Introduction}
Declarative specification languages have proven to be useful in a
variety of applications, however sometimes parts of specifications
contain duplicate information.  This commonly occurs when different
instantiations are needed of an abstract concept.  For example, in an
application, we may have to assert of multiple relations that they are
an equivalence relation, or multiple relations of which we need to
define their transitive closure. In most current logics, the
constraints (e.g., reflexivity, symmetry and transitivity) need to be
reasserted for each relation. 

In the early days of programming,
imperative programming languages suffered from a similar situation
where code duplication was identified as a problem. The first solution
proposed to this was the use of macros, where a syntactical
replacement was made for every instantiation of the macro. 
For specification languages, the analog for macros was introduced (e.g. in ASP), most often called \emph{templates}. These allow us to define a concept and instantiating it multiple times, without making the language more computationally complex. Asserting that the two relations \lstinline{P} and \lstinline{Q} are equivalence relations could be done using a template \texttt{isEqRelation} as follows:
\begin{lstlisting}[caption={This example defines an equivalence relation},label={lis:EquivRel},basicstyle={\small\ttfamily},mathescape=true]
{isEqRelation(F) $\leftarrow$ 
	$\forall$a : F(a,a).
	$\forall$a,b : F(a,b) $\Leftrightarrow$ F(b,a).
	$\forall$a,b,c : (F(a,b) $\land$ F(b,c)) $\Rightarrow$ F(a,c).
}
isEqRelation(P) $\land$ isEqRelation(Q).
\end{lstlisting}
In existing treatments of templates, their semantics is given in a
transformational way, essentially by translating them away as if they
were macros.  This appraoch has its limitations. An intellectually
more gratifying view, certainly in a declarative setting, is that
templates are higher order definitions. This allows for a much more
general treatment.  In some interesting cases, these higher order
template definitions are recursive (see Example~\ref{lis:game}). In others,
like the template symbol $tc(P,Q)$ specifying $P$ as the transitive
closure of $Q$, the definiens is itself an inductive definition (see
Example~\ref{lis:TransClos}) nested in the template definition of $tc$.

The goal of this work is to introduce a declarative template mechanism
for the language \fodot. This logic posesses an expressive first order
definition construct in the form of rules under well-founded semantics
which was shown suitable to express informal definitions of the most
common types \cite{KR/DeneckerV14}. We want to extend \fodot's
definition construc of to nested higher order definitions.

In the first part of this paper, we present a compositional framework for building an infinite class of logics. This framework specifies a principled way for  building rule formalisms under well-founded and stable semantics from arbitrary logics, and ways to  compose and nest arbitrary language constructs including higher order symbols, rule sets and aggregates.  In the second part we use this framework to build a template formalism.  As a last contribution we show that under suitable conditions,  the standard approach of templates as rewriting macros also works in this formalism, thus recovering the results of existing approaches.


\section{Related Work}
Abstraction techniques have been an important area of research since
the dawn of programming \cite{DBLP:journals/software/Shaw84}. 
Popular programming languages such as
C++ consider templates as a keystone for abstractions
\cite{musser2009stl}.
Within the ASP community, work by Ianni et al. \cite{nmr/IanniIPSC04}
and Baral et al. \cite{DBLP:conf/iclp/BaralDT06} introduced concepts to support
composability, called templates and macros respectively. The key idea
is to abstract away common constructs through the definition of generic `template'
predicates. These templates can then be resolved using a rewriting
algorithm. 

More formal attempts at introducing more abstractions in ASP were made.
Dao-Tran et al. introduced modules which can be used in
similar ways as templates \cite{DBLP:conf/iclp/Dao-TranEFK09} but has the disadvantage
that his template system introduces additional computational complexity,
so the user has to be very careful when trying to write an efficient specification.

Previously, meta-programming \cite{abramson1989meta} has also been used to introduce abstractions, for example in systems such as \hilog\ \cite{chen1993hilog}.
One of \hilog s most notable features is that it combines a higher-order syntax with a first-order semantics.
\hilog s main motivation for this is to introduce a useful degree of second order yet remain decidable.
While decidability is undeniably an interesting property, the problem of decidability already arises in logic programs under well-founded or stable semantics, certainly with the inclusion of inductive definitions: the issue of undecidability is not inherent to the addition of template behavior.
As a result, in recent times deduction inference has been replaced by various other, more practical inference methods such as model checking, model expansion, or querying.
Furthermore, for practical applications, we impose the restriction of stratified templates for which an equivalent first-order semantics exists.

An alternative approach is to see a template instance as a
call to another theory, using another solver as an
oracle. An implementation of this approach exists
in HEX \cite{DBLP:conf/inap/EiterKR11}. This implementation however suffers
from the fact that the different calls occur in different processes.
As a consequence, not enough information is shared which hurts the
search. This is analog to the approach presented in \cite{DBLP:conf/frocos/TasharrofiT11},
where a general approach to modules is presented. A template would
be an instance of a module in this framework, however the associated algebra
lacks the possibility to quantify over modules.

Previous efforts where made to generalize common language concepts, such as the work by Lifschitz \cite{iclp/Lifschitz99} who extended logic programs to allow arbitrary nesting of conjunction $\land$, disjunction $\lor$ and negation as failure in rule bodies. The nesting in  this paper is of very different kind,  by allowing the full logic, including definitions itself, in the body.

\section{ Preliminaries}

\newcommand{\Three}{{{\mathcal T}hree}}
\newcommand{\Two}{{{\mathcal T}wo}}

\newcommand{\natnrs}{{\mathbb N}}

\newcommand{\ty}{\tau}
\newcommand{\dt}{\delta}

\newcommand{\ol}[1]{{\widetilde{#1}}}

\newcommand{\OO}{{\mathcal O}}
\newcommand{\II}{{\mathcal I}}
\renewcommand{\I}{{I}}
\newcommand{\restr}[2]{{#1|_{#2}}}
\newcommand{\IInt}{{\mathcal I}nt}
\newcommand{\Int}{{I}nt}

\newcommand{\XXX}{{\mathcal X}}
\newcommand{\SSS}{{\mathcal S}}

\newcommand{\ass}{{\upsilon}}
\newcommand{\assw}[2]{{#1}^{#2}}
\newcommand{\assv}[2]{{#1}^{\ass:#2}}
\newcommand{\spass}{{\bf \mathit s}}
\newcommand{\spv}[2]{{#1}^{\spass:#2}}
\newcommand{\kass}{{\bf \mathit k}}
\newcommand{\kv}[2]{{#1}^{\kass:#2}}

\newcommand{\defin}[1]{{\left\{\begin{array}{l}#1
    \end{array}
    \right\}}}

\newcommand{\compop}{{\sim}}



\paragraph{Symbols.} We assume an infinite supply of (typed)
symbols. A vocabulary $\voc$ is a set of (typed) symbols. For each
symbol $\sigma$, $\ty(\sigma)$ is its type.  For a tuple
$\bar{\sigma}$, $\ty(\bar{\sigma})$ denotes the tuple of types. 

An untyped logic is one with a single type. But for the purposes of
this paper, it is natural to use at least a simple form of typing,
namely to distinguish between first order symbols and the second order
(template) symbols. We distinguish between base types (some of which may
be interpreted, e.g., $\bool, \int$) and composite types.  A simple
type system that suffices for this paper consists of the following
types:
\begin{itemize}
\item base types $\dt$ and $\bool$; $\dt$ represents the domain; 
\item first order types: $n$-ary predicate types $\dt^n\ra\bool$ and function types  $\dt^n\ra\dt$. 
As usual, propositional symbols and constants are predicate and function symbols of arity $n=0$.
\item second order types: $n$-ary predicate types $(\ty_1,\dots,\ty_n)\ra\bool$ with each $\ty_i$ a first order type or $\dt$. 
\end{itemize}
This is the type system that we have in mind in this paper. It
suffices to handle untyped first order logic and second order
predicates (no second order functions are needed). However, the
framework below is well-defined for much richer type systems
(including higher order types, type theory).

\ignore{For reasons of simplicity, function symbols are not considered
  in this paper. However, the framework can be easily extended with
  them. Moreover, for logics such as FO, n-ary function symbols can be
  emulated using their n+1-ary graph predicate symbols.}

\paragraph{(Partial) values.}

For interpreted base types, there is a fixed domain of values. E.g., the domain of the boolean type $\bool$ is $\Two=\{\Tr,\Fa\}$. For other base types $\ty$, the domain  of values is chosen freely. For composite types, the set of values is constructed from the values of the base types. 

For the simple type system above, the values of all types are determined by the choice of the domain associated with $\dt$. For any domain $D$, we can define the domain  $\ty^D$ of any type $\ty$ as follows:
\begin{itemize}
\item  $\dt^D=D$, $\bool^D=\Two$
\item first order predicates: $(\dt^n\ra\bool)^D$ is the set of all functions from $D^n$ to $\Two$ (or equivalently,  the set of all subsets of $D^n$).  For first order functions,  $(\dt^n\ra\dt)^D$ is the set of all functions from $D^n$ to $D$.
\item second order predicates:  $((\ty_1,\dots,\ty_n)\ra\bool)^D$  is the set of all functions from  ${\ty_1}^D\times\dots\times{\ty_n}^D$ to $\Two$.
\end{itemize}

To define the semantics of inductive definitions, partial values for
predicates are essential (since only predicates are defined in the
logics of this paper, we do not introduce partial values for
functions).  A {\em partial set} on domain $D$ is a function from $D$
to $\Three=\{\Tr,\Un,\Fa\}$. A {\em partial} value of a predicate type
$\ty'=(\bar{\ty}\ra\bool)$ in domain $D$ is a {\em partial set} with
domain $\bar{\ty}^D$. $\Three$ extends $\Two$ and is equipped with two
partial orders: the truth order $\leq$ is the least partial order
satisfying $\Fa\leq \Un\leq\Tr$, the precision order $\leqp$ the least
partial order satisfying $\Un\leqp\Fa, \Un\leqp\Tr$. The orders $\leq$
and $\leqp$ on $\Three$ are pointwise extended to partial sets.  $\Un$
is seen as an approximation of truth values, not as a truth value in
its own right. A partial set that is maximally precise has range
$\Two$ and is called {\em exact}.  A partial set $\mathcal S$ is seen as an approximation of any exact set $S$ for which $\mathcal{S} \leqp S$.

\paragraph{(Partial) Interpretations.}

A partial $\voc$-interpretation $\II$ consists of a suitable domain
$\ty^\II$ for every type $\ty$ in $\voc$ (which is the set of partial
sets on ${\ty_d}^\II$ in case $\ty$ is a predicate type with domain
type $\ty_d$), and for every symbol $\sigma\in\voc$ of type $\ty$ a
value $\sigma^\II\in\ty^\II$. An exact $\voc$-interpretation is one
that assigns exact values. The class of partial
$\voc$-interpretations is denoted $\IInt(\voc)$; the class of exact
$\voc$-interpretations is $\Int(\voc)$.

The precision order $\leqp$ and truth order $\leq$ are extended to
partial interpretations in the standard way: $\II\leqp\II'$ if $\II,
\II'$ interpret the same vocabulary $\voc$, have the same values for
all types and non-predicate symbols, and $P^\II\leqp P^{\II'}$ for
every predicate symbol $P\in\voc$.  Likewise for the truth order
$\leq$.  We use $\II$ to denote a partial interpretation (which may be
exact) and $\I$ to denote an exact interpretation. 

The restriction of a $\voc$-interpretation $\II$ to
$\voc'\subseteq\voc$ is denoted as $\restr{\II}{\voc'}$. If $\II$ is a
partial $\voc$-interpretation, $\sigma$ a symbol (that might not
belong to $\voc$) and $v$ a well-typed value for $\sigma$,
then $\II[\sigma:v]$ is the $(\voc\cup\{\sigma\})$-interpretation
identical to $\II$ except that $v$ is the value of $\sigma$.

Given an interpretation $\II$ of at least the types of $\voc$, a {\em
  domain atom} of an $n$-ary predicate symbol $P\in\Sigma$ of type
$\bar{\ty}\ra\bool$ in $\II$ is a pair $(P,\ddd)$ where $\ddd\in
\bar{\ty}^\II$. It is denoted as $P(\ddd)$.  If $\II$ interprets $P$,
a domain atom $P(\ddd)$ has a truth value $P(\ddd)^\II = P^\II(\ddd)$.

For any $v\in\Three$ and set $X$ of domain atoms of partial
interpretation $\II$, we denote $\II[X:\Tr]$ the interpretation
identical to $\II$ except that each $A\in X$ is true; similarly for
$\II[X:\Un], \II[X:\Fa]$.  We may concatenate such notions and write
$\II[X:\Un][Y:\Fa]$, with the obvious meaning (first revising $X$,
next revising $Y$).

\paragraph{Logics $\Log^\ass$.}

A logic is specified as a pair $(\Log,\ass)$ (denoted $\Log^\ass$)
such that $\Log$ is a function mapping vocabularies $\voc$ to sets
$\Log(\voc)$ of expressions over $\voc$, and $\ass$ is a two-valued or
three-valued truth
assignment. 
An expression $\varphi$ of $\Log(\voc)$ has free symbols in $\voc$; it
may contain other symbols provided they are bound by some scoping
construct in a subexpression of $\varphi$ (e.g., a quantifier). If
$\voc\subseteq\voc'$, then $\Log(\voc)\subseteq\Log(\voc')$.

A (three-valued) truth assignment $\ass$ maps tuples $(\varphi,\II)$
where $\II$ interprets all free symbols of $\varphi$, to
$\Three$. This function satisfies the following properties: (1) if
$\varphi\in\Log(\voc)$, $\voc\subseteq \voc'$ and $\II$ is a
$\voc'$-interpretation, then
$\assv{\varphi}{\II}=\assv{\varphi}{\restr{\II}{\voc}}$; (2)
exactness: $\assv{\varphi}{\I}\in\Two$ for every exact interpretation
$\I$; (3) $\leqp$-monotonicity: if $\II\leqp\II'$ then
$\assv{\varphi}{\II}\leqp\assv{\varphi}{\II'}$. A two-valued truth
assignment $\ass$ is defined only for exact interpretations and
satisfies (1) and (2). 

\begin{defn}
We say that two formulas $\varphi_{1}$ and $\varphi_{2}$ over $\voc_{1}$ and $\voc_{2}$ respectively are $\voc$-\emph{equivalent}, with $\voc \subseteq (\voc_{1} \cap \voc_{2})$, if for any interpretation $\I$ over $\voc$, there exists an expansion $\I_{1}$ to $\voc_{1}$ for which $\assv{\varphi}{\I_{1}}=\Tr$ iff there exists an expansion $\I_{2}$ of $\I$ to $\voc_{2}$ for which $\assv{\varphi}{\I_{2}}=\Tr$. If in addition  $\voc_1=\voc_2=\voc$, we call  $\varphi_{1}$ and $\varphi_{2}$ equivalent; hey have the same truth value in all $\voc$-interpretations.
\end{defn}

\ignore{
\paragraph{Views.} A three-valued semantics for $\Log$ induces a
two-valued one. In particular, $\ass$ induces a standard satisfaction
relation: $\I\models\varphi$ if $\assv{\varphi}{\I}=\Tr$; in turn this
defines the entailment relation: $\varphi \models \psi$ if
$\I\models\varphi$ implies $\I\models\psi$. As common in Tarskian
semantics, we view an exact interpretation as a mathematical abstraction
of a states of affair and $\assv{\varphi}{\I}$ as whether this state of
affairs satisfies $\varphi$. Partial interpretations are not
representations of states of affairs but approximations of these. 
}

\section{Well-founded and stable semantics for $\Log$-rule sets}

\label{SecRules}

In this section, we show that for each logic $\Log^\ass$ with a
three-valued truth assignment $\ass$, it is possible to define a rule
logic under a well-founded and under a stable semantics. Let
$\Log^\ass$ be  such a logic.

\begin{definition}
An $\Log$-rule over $\voc$ is an expression $\forall \xxx (P(\xxx)\rul\varphi)$
with $P$ a predicate symbol in $\voc$, $\xxx$ a tuple of ``variable''
symbols and $\varphi\in \Log(\voc\cup\{\xxx\})$. 
An $\Log$-rule set over $\voc$ is a set of $\Log$-rules over
$\voc$. Rule sets will be denoted with $\D$.
\end{definition}
The set $\defp{\D}$ is the set of predicate symbols $P\in\voc$ that
occur in the head of a rule.  $\pars{\D}$ is the set of all other
symbols that occur in $\D$. Elements of $\defp{\D}$ are called {\em
  defined} symbols, the other ones are called {\em parameters} of
$\D$. 

\begin{definition} 
  A context $\OO$ of a $\Log$-rule set $\D$ is a
  $\voc\setminus\defp{\D}$-interpretation.
\end{definition}

For a given context $\OO$, the set $\{\II \mid
\restr{\II}{\pars{\D}}=\OO\}$ of partial $\voc$-interpretations
expanding $\OO$ is isomorphic to the set of partial sets of domain
atoms of $\defp{\D}$ in $\OO$. Thus, given $\OO$, a partial set of
domain atoms specifies a unique partial interpretation $\II$ expanding
$\OO$ and vice versa.

\ignore{
With this convention in place, we can define the immediate consequence operator for $\D$ in context $\OO$. 
\begin{definition}
Given a $\Log$-definition $\D$ and context $\OO$, the immediate consequence operator $\ICO{\D}{\OO}$ is a mapping of  expansions $\OO$ such that if $\II'=\ICO{\D}{\OO}(\II)$ if for every defined domain atom $P(\ddd)$:
\[\assv{P(\ddd)}{\II'}= Max_\leq(\{\assv{\varphi[\ddd]}{\II} \mid (\forall \xxx: P(\xxx)\rul\varphi) \in \D\}\]
\end{definition}

\begin{proposition}
$\ICO{\D}{\OO}$ is $\leqp$-monotone and preserves exactness.
\end{proposition}
}

We call a set of domain atoms  a $\Tr$-set, respectively $\Un$-set, $\Fa$-set of partial interpretation $\II$ if its elements have truth value $\Tr$, respectively $\Un$, $\Fa$ in  $\II$. 

\begin{definition}
A partial interpretation $\II$ is  closed under $\D$ if for any domain atom $P(\ddd)$ and rule $\forall \xxx (P(\xxx)\rul \varphi) \in \D$, if $\assv{\varphi[\ddd]}{\II} =\Tr$ then $\assv{P(\ddd)}{\II}=\Tr$.
\end{definition}

\begin{definition}
An unfounded set of $\D$ in   $\II$  is a $\Un$-set $U$ of defined domain atoms in $\II$, for which every atom $P(\ddd)\in U$ and rule  $\forall \xxx (P(\xxx)\rul \varphi) \in \D$,  $\assv{\varphi[\ddd]}{\II[U:\Fa]} =\Fa$.
\end{definition}

\begin{definition}\label{DefPartialStable}
A partial interpretation $\II$ extending context $\OO$ is a {\em partial stable interpretation} of $\D$ if 
\begin{enumerate}[leftmargin=*]
\item for each domain atom $P(\ddd)$, $P(\ddd)^\II=Max_\leq\{\assv{\varphi[\ddd]}{\II} \mid \forall \xxx(P(\xxx)\rul\varphi) \in \D\}$;

\item (prudence) there exists no non-empty $\Tr$-set $T$ and no (possibly empty) $\Un$-set $U$ of $\II$ such that $\II[T:\Un][U:\Tr]$ is closed under $\D$;

\item (braveness) the only unfounded set of $\D$ in $\II$ is $\emptyset$.
\end{enumerate}
\end{definition}

\begin{definition}
We call a partial interpretation $\II$ a well-founded interpretation of $\D$ if $\II$ is the $\leqp$-least partial stable model $\II'$ of $\D$ such that $\restr{\II'}{\pars{\D}}=\restr{\II}{\pars{\D}}$. 
\end{definition}

\begin{definition}
We call an (exact) interpretation $\I$ a stable interpretation of $\D$ if $\I$ is an exact partial stable model of $\D$. 
\end{definition}
Given that a stable $\I$ has only the empty $\Un$-set, conditions (2) and (3) simplify to that there is no non-empty $\Tr$-set $T$ of $\II$ such that $\II[T:\Un]$ is closed under $\D$.

\ignore{
\footnote{The definition of well-founded interpretation of
  a definition $\D$ is non-constructive. A constructive definitions is
  as follows.   Let us assume, without loss of generality, that a definition
  $\D$ consists of exactly one rule $\forall \xxx
  (P(\xxx)\rul\varphi_P[\xxx])$ for each defined predicate $P$, and
  for a domain atom $A$ of the form $P(\ddd)$, let $\varphi_A$ denote
  $\varphi_P[\ddd]$. We define that interpretation $\II'$ is a
  $\D$-refinement of interpretation $\II$ if there is a $\Un$-set $U$
  of $\II$ such that one of the following conditions is satisfied:
  \begin{itemize}
  \item for each $A\in U$, ${\varphi_A}^\II=\Tr$ and
    $\II'=\II[U:\Tr]$; or
  \item for each $A\in U$, ${\varphi_A}^{\II'}=\Fa$ and
    $\II'=\II[U:\Fa]$.
  \end{itemize}
  Then $\II$ is a well-founded interpreteation if it is the limit of
  each maximally long possibly transfinite sequence
  $\struct{\II_\alpha}_{0\leq\alpha}$ of partial structures that is
  strictly increasing in precision, such $\II_0$ is identical to $\II$ except that all defined atoms are $\Un$, and each next sequence element $\II_{\alpha+1}$  is a refinement of $\II_{\alpha}$. }
}

\begin{proposition}
Let $\varphi,\varphi'$ be equivalent under $\ass$ (same truth value in all partial interpretations). Then substituting $\varphi$ for $\varphi'$ in the body of a rule of $\D$ preserves the class of partial stable (hence, well-founded and stable) interpretations. 
\end{proposition}
\begin{proof} This is trivial, since the conditions of partial stable interpretation are defined in terms $\ass$, which cannot distinguish $\varphi$ from $\varphi'$. 
\end{proof}


\paragraph{Two logics.}  Using the above two concepts we define two
rule logics. Expressions in both logics are  the same: finite sets of rules. 

\begin{definition}
For  logic $\Log^\ass$, we define  logic $R(\Log^\ass)^w$ where $R(\Log^\ass)(\voc)$ is the collection of finite rule sets over $\voc$ and $w$ the two-valued truth assignment  defined as $\assw{\D}{w:\I}=\Tr$ if $\I$ is an exact well-founded interpretation of $\D$ and $\assw{\D}{w:\I}=\Fa$ otherwise.
\end{definition}

\begin{definition}
For logic $\Log^\ass$, we define the logic $R(\Log^\ass)^{st}$ where $R(\Log^\ass)(\voc)$ is  as above and $st$ is the two-valued truth assignment defined as $\assw{\D}{st:\I}=\Tr$ if $\I$ is an exact stable interpretation of $\D$ and  $\assw{\D}{st:\I}=\Fa$ otherwise. 
\end{definition}

For the logic FO$^\kass$, with FO first order logic and $\kass$ the
3-valued Kleene truth assignment \cite{Kleene52}, the rule formalism
$R(FO^\kass)^w$ corresponds to the (formal) definitions in the logic
FO(ID) \cite{Denecker:CL2000,tocl/DeneckerT08} while the formalism
$R(FO^\kass)^{st}$ corresponds to the rule formalism in the logic
ASP-FO \cite{DeneckerLTV12}.

In \cite{KR/DeneckerV14},
the relation between the main forms of (informal) {\em definitions} found in
mathematical text, and rule sets in $R(FO^\kass)^w$ was analyzed.   Not
all rule sets of $R(FO)$ express sensible (informal) definitions, but
for those that do, the well-founded interpretations are exact and
correctly specify the defined sets. Therefore, a rule set $\D$ was called a
{\em paradox-free} or {\em total definition} in an exact context $\OO$ if
its well-founded interpretation expanding $\OO$ is exact. For
paradox-free rule sets, $w$ and $st$ coincide. Important classes of
rule sets are always paradox-free: non-recursive, monotone inductive
rule sets, and rule sets by ordered or iterated induction over some
well-founded induction order as defined in \cite{KR/DeneckerV14}.  

\ignore{
\paragraph{Remark} If the well-founded model of $\D$ in $\OO$ is
exact, then it is the unique stable model of $\OO$. It follows that
for rule sets that express sensible definitions, well-founded and
stable semantics coincide. Hence, the discussion between these
semantics is not so much what is the right meaning of rule sets as it
is about what rule sets are sensible.  In the logic FO(ID), the more
restrictive view is that rule sets with non-exact well-founded model
(in some context $\OO$) contain a semantic error. But for rule sets
that are sensible, the two semantics are the same. WEG??
}

 \ignore{
  \footnote{More precisely, partial structures $\II$ extending $\OO$
    correspond to pairs $(a,b)$ where $a$ is the set of true domain
    atoms, and $b$ the set of non-false atoms. With $\D$, we can
    define an approximator $A$ mapping pairs $(a,b)$ to the pair
    $(a',b')$ where $a'$ is the set of domain atoms $P(\ddd)$ with a
    rule $\forall\xxx(P(\xxx)\rul\varphi)\in\D$ such that
    $\varphi[\ddd]^\II=\Tr$ and $b'$ is the set of domain atoms
    $P(\ddd)$ with a rule $\forall\xxx(P(\xxx)\rul\varphi)\in\D$ such
    that $\varphi[\ddd]^\II\neq\Fa$. This is an approximator.  A
    partial stable fixpoing of $A$ is a pair $(a,b)$ such that 1)
    $A(a,b)=(a,b)$, 2) (prudence) if for some $z\leq b$, $A(z,b)_1\leq
    z$, then $a\leq z$, and 3) (braveness) if for some $a\leq z\leq
    b$, $A(a,z)_2\leq z$, then $b\leq z$. NU ZOU IK ER NOG EEN
    ARGUMENT KUNNEN BIJPLAKKEN DAT LAAT ZIEN DAT DEZE CONDITIES
    OVEREENKOMEN.}  }

\ignore{
\paragraph{Views.} In \cite{}, it was pointed that  inductive definitions in mathematical text are used to define formal objects but are not formal objects themselves. Often they are phrased as sets of informal rules.
\begin{example}
\begin{definition} For any node $A$ of graph $G$, we define the set $R_A$ by induction:
\begin{itemize}
\item Node $A\in R_A$.
\item  $x\in R_A$ if for some $y$, $(x,y)$ is an edge of $G$ and $y\in R_A$.
\end{itemize}
\end{definition}
This informal definition defines the set of nodes that can reach $A$ in $G$. Here $G$ and $A$ are parameters. 
\end{example}
Such rule sets often can easily be expressed  as formal rules with FO bodies.
\begin{example}
\[ \defin{R(A)\rul\\ \forall x(R(x)\rul\exists y(G(x,y)\land R(y)))}\]
\end{example}
The main claim is an empirical science thesis: if 
\begin{quote}
If $\D$ is a  monotone inductive definition or an inductive definition by well-founded induction over some well-fopunded induction order, and $\D$ is a faithful translation of it, then in any suitable context $\OO$, the well-founded interpretation $\II$ of $\D$ in $\OO$ will be two-valued and the values  $P^\II$ of defined symbols of $\D$ will be the corresponding sets defined by $D$. If $\II$ is not exact, then $\D$ is not a sensible definition in $\OO$. 
\end{quote} 
It is a thesis, hence cannot be proven.  In Poppers tradition
of empirical science, it can be easily falsified. One counterexample
suffices: one sensible informal inductive definition of the mentioned
types (monotone or ordered) in context $\OO$, with a faithful formal
translation $\D$, such that either the well-founded interpretation of
$\D$ in $\OO$ is not exact, or the value $P^\II$ is not the defined
value.  For example, one graph $G$ with node $A$ such that $R^\II$ is
not the set of nodes that can reach $A$ in $G$ would falsify the
thesis. In this case, this is impossible since $\D$ is a monotone rule
set, $\II$ is the least set satisfying it, and this is the defined set
of the informal definition. But perhaps other counterexamples exist.
So far, the thesis stands. For an argument why , see \cite{}.
}

\section{Compositional framework for building logics with definitions}

This section effectively defines an infinite collection of logics.  We
define compositional constructs which add new expressions, such as
definitions, to an existing logic.  By iterating such extension steps,
these constructs can be nested.


\ignore{
The logics constructed in this section are formal languages $\Log$
equipped with a class of partial and exact interpretations, with a
three-valued truth assignment $\ass$ satisfying $\leqp$-monotonicity
and exactness. By limiting $\ass$ to exact interpretations, a
two-valued Tarskian model semantics is induced, a satisfaction
relation $\models$ between (exact) interpretations and
$\Log$-expressions or theories, an entailment relation $\models$
between pairs of theories or expressions of $\Log$. }

\subsection{Approximating boolean functions} 
\label{ssec:extending:boolean:functions}
We frequently need to extend a boolean function defined on a domain $X$
of exact values (e.g., $\Two$, exact sets, exact interpretations, or
tuples including these) to the domain $\mathcal X,\leqp$ of partial
values. Examples are the boolean functions associated with connectives
$\neg, \land,\dots$, or the truth assigments $w$ and $st$ of $R(\Log^\ass)$
defined on $\Int(\Sigma)$. Given such a function $F: X\ra\Two$, we
search for an approximation ${\mathcal F}:{\mathcal X}\ra \Three$ such
that:
\begin{itemize}
\item $\leqp$-montone: if ${\mathbf x}\leqp{\mathbf  y} \in {\mathcal X}$ then ${\mathcal F}({\mathbf x})\leqp {\mathcal F}({\mathbf  y})$;
\item exact and extending: for  $x\in X$,  ${\mathcal F}(x)=F(x)$. 
\end{itemize}

\begin{defn} We  define the ultimate approximation $\ol{F}: {\mathcal X}\ra \Three$ of $F$ by defining $\ol{F}(\mathbf x)=\glb_{\leqp}\{F(x)\mid {\mathbf x}\leqp x\} \in\Three$.
\end{defn}
\begin{proposition}
$\ol{F}$ is the most precise $\leqp$-monotone exact extension of $F$. 
\end{proposition}
\begin{proof} $\leqp$-monotonicity follows from the transitivity of $\leqp$. Exactness, from the fact that elements of $\Int(\voc)$ are maximally precise. That $\ol{F}$ is the most precise approximation of $F$ is clear as well.
\end{proof}
Several important examples follow. For a standard connective
$c\in\{\land, \lor, \neg, \mim, \equi\}$ with corresponding boolean function 
${\mathbf c}:\Two^n\ra\Two$,
the function $\ol{\mathbf c}:\Three^n\ra \Three$ is the three-valued truth
function used in the Kleene truth assignment $\kass$.

The semantics of quantifiers $\forall, \exists$ and generalized
quantifiers such as aggregates are given by functions on sets (or
tuples including sets). E.g., for quantification over domain $D$ these
are the boolean functions $\forall_D, \exists_D$ defined
$\forall_D(S)=\Tr$ iff $D\subseteq S$, and $\exists_D(S)=\Tr$ iff
$D\cap S\neq \emptyset$. Two commonly used numerical aggregate
functions are cardinality $\#$ and $sum$ (the latter mapping (finite)
sets $S$ of tuples $\bar{d}$ to $\sum_{\bar{d}\in S}d_1$). For every
numerical aggregate function $Agg$ and boolean operator $\compop\in
\{=,<,>\}$ on numbers, the boolean function $Agg_\compop$ maps tuples
$(S,n)$ to $\Tr$ iff $Agg(S)\compop n$.

For all these higher order boolean functions $F$, $\ol{F}$ is the most
precise approximation on three-valued sets. The three-valued
aggregate functions $\ol{Agg}_{\compop}$ were introduced originally in
\cite{tplp/PelovDB07} to define stable and well-founded semantics for
aggregate logic programs.  The functions $\ol{\forall}_D,
\ol{\exists}_D$ are used in the Kleene truth assignment $\kass$: let $D$
be $\ty(x)^\II$, and ${\mathcal S}=\{(d,\kv{\varphi}{\II[x:d]})\mid d\in D\}$, i.e. the three valued set mapping
domain elements $d\in D$ to $\kv{\varphi}{\II[x:d]})$, one defines
$\assw{(\forall x\ \varphi)}{\kass:\II}=\ol{\forall}_D({\mathcal S})
=Min_\leq\{{\mathcal S}(d)\mid d\in D\} =
Min_\leq\{\kv{\varphi}{\II[x:d]})\mid d\in D\}$.

For any two-valued truth assigment $\ass$ on $\Log$, $\ol{\ass}$ is a
sound three-valued truth assignment.  In case of FO and its truth
assignment $\ass$, $\ol{\ass}$ was introduced in
\cite{jp/vanFraassen66} where it was called the supervaluation
$\spass$. $\spass$ is not truth functional, for if $p^\II=q^\II=\Un$,
then $\assw{(p\lor\neg p)}{\spass:\II}=\Tr\neq\assw{(p\lor
  q)}{\spass:\II}=\Un$ while  the components of the two
disjunctions have the same supervaluation.  A truth-functional
definition of a three-valued truth assignment is obtained by using the
ultimate approximations of the boolean functions associated to
connectives and quantifiers. This yields exactly the Kleene truth
assigment $\kass$. It is $\leqp$-monotone, exact and extending, and
strictly less precise than $\spass$ as can be seen from
$\assw{(p\lor\neg p)}{\kass:\II}=\Un$: the supervaluation ``sees'' the
logical connection between $p$ and $\neg p$ in this tautology while
$\kass$ does not. 

Other applications serve to extend, for arbitrary logic $\Log^\ass$,
the two-valued wellfounded and stable truth assignments $w$ and $st$
on $R(\Log^\ass)$ to three-valued extensions $\ol{w}, \ol{st}$. Here,
it holds that $\assw{\D}{\ol{w}:\II}=\Tr$ (respectively $\Fa$) if every
(respectively, no) instance $\I$ of $\II$ is a well-founded
interpretation of $\D$.

\subsection{Composing logics by combining logic constucts}
\label{sec:Composing}
A standard way of defining the syntax of a logic is through a set of
often inductive syntactical rules, typically described in Backus Naur
Form (BNF). The truth assignment $\ass$ is then defined by recursion
over the structure of the expressions.  Below, we identify a language
construct $C$ with a pair of a syntactical and a semantical rule. The
rules below construct, for a language construct $C$, a new logic
$C(\Log^\ass)^{\ass'}$ with expressions obtained by applying $C$ on
subexpressions of $\Log^\ass$ and $\ass'$ a truth assignment for
$C(\Log)$. Afterwards, complex logics with multiple and nested
language constructs can be built by iterating these construction
steps.
  
\begin{itemize}
\item Atom$^1$ and Atom$^2$: for first order predicates $p$  and second order ones $P$  respectively. $\ttt$ is a tuple of terms, $\xxx$ of variables. 
\[Atom^1 ::= p(\ttt)\text{ where } p(\ttt)^{\ass':\II}= p^\II(t_1^\II,\cdots,t_n^\II)\]
\[Atom^2::= P(\xxx)\text{ where } P(\xxx)^{\ass':\II}= P^\II(x_1^\II,\cdots,x_n^\II)\]
\item N-ary connectives $c  \in   \{\neg, \land, \lor, \Rightarrow, \Leftrightarrow\}$, 
\[c({\Log}^{\ass}) ::= c(\alpha_1,\dots,\alpha_n) \text{ where } \assw{c(\alpha_1,\dots,\alpha_n)}{\ass':\II}=\ol{\mathbf c}(\assv{\alpha_1}{\II},\cdots,\assv{\alpha_n}{\II})\]
\item Generalized quantifiers $C \in \{\forall, \exists, Agg_\compop\}$. 
Below, $C\langle \xxx,\alpha,z\rangle$ denotes the syntactic expression, e.g., $\forall\langle x,\alpha\rangle$ is $\forall x\ \alpha$; $Agg_\compop\langle \xxx,\alpha,z\rangle$ is $Agg\{\xxx:\alpha\}\compop z$. \[C({\Log}^{\ass}) ::= C\langle\xxx,\alpha,z\rangle \text{ where } \assw{(C\langle \xxx,\alpha,z\rangle)}{\ass':\II}=\ol{\mathcal C}(\{(\ddd,\assv{\alpha}{\II[\xxx:\ddd]})|\ddd\in\ty(\xxx)^\II\},z^\II)\]

\item Definitions as rule sets ($R_w$) (similarly, one could define $R_{st}$):   
\[R_w({\Log}^{\ass}) ::= R(\Log) \text{ where }\assw{\D}{\ass':\II}=\assw{\D}{\ol{w}:\II}\]
where $w$ is the  well-founded assignment of $R(\Log^\ass)${ as defined in Section 4.}
\end{itemize}

\paragraph{Building logics.} \label{sub:soidDef} Using the above rules of language constructs, an (infinite)
class of logics with three-valued semantics can be built. Moreover, every combination of the above rules gives rise to a valid three-valued truth assignment.
\begin{proposition}
 Every (sub)set of the above language constructs (possibly closed under recursive application) defines a logic with a proper three-valued truth assignment (i.e. it is $\leqp$-monotone, exact and extending).
\end{proposition}

For example, given a logic $\Log^\ass$, we define ${\Log'}^{\ass'}=
R(\Log^\ass)^{\ol{w}}$ by one application of $R_w$ on $\Log^\ass$.  By
iterating $R_w$, logics $(R(\Log^\ass)^{\ol{w}})^n$ with nested
definitions are built.  Every BNF in terms of the above language
constructs now implicitly defines a three-valued logic.  The
definition of first order logic FO$^\kass$ with $\kass$ the standard
three-valued Kleene truth assignment, can be descibed in BNF or more
compactly as $\{Atom^1,\land,\lor,\neg,\forall,\exists\}^*$ (here $*$
indicates recursive application of the construction rules). The logic
FO(ID) defined in \cite{tocl/DeneckerT08} is the union of logics
$FO^\kass$ and $R(FO^\kass)^{\ol{w}}$. A further extension is the new
logic $FO(ID^*) = \{Atom^1,\land,\lor,\neg,\forall,\exists,R_w\}^*$
which has definitions nested in formulas and definition rule bodies. A
logic in which templates can easily be embedded is $SO(ID^*) =
\{Atom^1,Atom^2,\land,\lor,\neg,\forall,\exists,R_w\}^*$. It is a second
order extension of FO(ID) which allows for nesting of definitions in rule bodies.

\ignore{
If we combine this technique with the one of the previous section, we
find that a suitable logic $\Log$ induces a logic $\Log\cup
D(\Log^\ass)$. An example of such a logic that appeared in the
literature is the logic FO(ID) \cite{2000,2008}. Here $\Log$ is
classical first order logic FO and $FO(ID) = FO\cup D(FO)$. Also the
logic $FO\cup St(FO)$ was introduced in \cite{ASP-FO} as the logic
ASP-FO. }

\section{Templates}
We envision  a library of application independent templates in
the form of second order definitions that encapsulate prevalent
patterns and concepts and that can be used as building blocks to
compose logic specifications. Below, we formally define the concepts
and show that non-recursive templates do not increase the
computational complexity of \foids and can be eliminated by a
rewriting process.

\subsection{Definition and usage}

We assume the existence of a set of template symbols. 
A template is a  context-agnostic  second order definition of template symbols. As such it should define and contain  only domain independent symbols: interpreted symbols and template symbols. A template might define a template symbol in terms of other template symbols and interpreted predicates, but not in terms of user-defined symbols. 

\begin{defn}
The template vocabulary $\Sigma_{Temp}$ is the vocabulary consisting of all interpreted symbols (such as arithmetic symbols) and all template symbols.
\end{defn}

\begin{defn}
A template is a second order definition $\D$ over $\Sigma_{Temp}$ such that  $\defined{\D}$ consists  of template symbols. 
\end{defn}
Thus the set of parameters $\pars{\D}$ of a template consists only of
interpreted symbols and template symbols.

The concepts used in Example \ref{lis:EquivRel} are now fully defined. Another common example is the template \lstinline{tc} expressing that Q/2 is the transitive closure of P/2, as shown in Example \ref{lis:TransClos}.
Note that this example cannot be written without a definition in the body of the template, so this further motivates our choice to allow definitions in the bodies of other definitions in our recursive construction of the logic $\soids$.
\begin{lstlisting}[caption={This template TC expresses that Q is the transitive
closure of P},label={lis:TransClos},mathescape=true]
{tc(P,Q) $\leftarrow$
	{Q(x,y) $\leftarrow$ P(x,y) $\lor$($\exists$ z: Q(x,z)$\land$Q(z,y))}.
}
\end{lstlisting}


Another notable aspect of this approach to templates is that recursive templates are well-defined. This enables us to write recursive templates, for example to define a range:
\begin{lstlisting}[caption={P is the range of integers from a to b},label={lis:recur},mathescape=true]
{range(P, a, b) $\leftarrow$
	{P(a).
	 P(x) $\leftarrow$ a < b $\land$ ($\exists$ Q : range(Q,a+1,b) $\land$ Q(x)).
	}
}
\end{lstlisting} 

It is possible to rewrite Example \ref{lis:recur} into a non-recursive template. Example \ref{lis:game} contains an example which is not rewritable in such a way.
\begin{lstlisting}[caption={\texttt{cur} is a winning position in a two-player game},label={lis:game},mathescape=true]
{win(cur,Move, IsWon) $\leftarrow$ IsWon(cur) $\lor$
	$\exists$ nxt : Move(cur,nxt) $\land$ lose(nxt,Move,IsWon).	
 lose(cur,Move, IsWon) $\leftarrow$ $\neg$IsWon(cur) $\land$ 
 	$\forall$ nxt : Move(cur,nxt) $\Rightarrow$ win(nxt,Move,IsWon).
}
	
\end{lstlisting}
This template defining $win$ and $lose$ by simultaneous definition, is a monotone second order definition and has a two-valued well-founded model. That it cannot be rewritten without recursion over second order predicates follows from the fact that deciding if a tuple belongs to a non-recursively defined second order predicate is in PH while deciding winning positions in generalized games is harder (if the polynomial hierarchy does not collapse) and this last problem corresponds to deciding elementship in the relation \lstinline{win} defined in  Example~\ref{lis:game}. 

\begin{definition} 
A template library $L$ is  a finite set of templates satisfying (1) every template is paradox-free; (2) every template symbol is defined in exactly one template; (3) the set of templates is hierarchically stratified: there is a strict order $<$ on template symbols such that for each $\D\in L$, if $P\in \defp{\D}, Q\in \pars{\D}$ then $Q<P$.
\end{definition}

\begin{proposition} For a template library $L$,  each interpretation $\struct$ not interpreting symbols of $\Sigma_{Temp}$ has a unique two-valued expansion $\struct'$ to $\Sigma_{Temp}$ that satisfies $L$.
\end{proposition}
\begin{proof} By induction on the hierarchy $<$ of $L$.
\end{proof}

\subsection{The {\boldmath$\Sigma_{Temp}$} vocabulary restriction}

The condition that templates should be built from $\Sigma_{Temp}$ and not from user-defined symbols is to ensure that templates are domain independent `drop-in' building blocks. 
%
This  restriction might seem too stringent, but we can show that many template definitions for which it does not hold, can be rewritten as an equivalent one for which it holds.

\newcommand{\ooo}{{\overline{o}}}
Let $\D$ be definition of second order predicates with $\defp{\D}\cap\Sigma_{Temp}=\emptyset$, and $\ooo$ the tuple of all free (user-defined) symbols of $\pars{\D}\setminus \Sigma_{Temp}$  (arranged in some arbitrary order). For such definitions, we define a templified version. For any rule or formula $\Psi$, we define $\Psi^\ooo$ to be  $\Psi$ except that every atom $P(\ttt)$ in $\Psi$ with $P\in\defp{\D}$ is replaced by $P'(\ttt,\ooo)$, with $P'$ a new symbol extending $P$ with new arguments corresponding to $\ooo$.

We say that a structure $\struct$ corresponds to $\struct'$  if $\struct, \struct'$   interpret the free symbols of $\D$, respectively those of $\D_{Temp}$, they are identical on shared symbols and  for each $P\in\defp{\D}$, $P^\struct=\left\{ \ddd | (\ddd,\ooo^\struct) \in {P'}^{\struct'}\right\}$. Note that for each $\struct'$ and  each value $\ddd_o$ for $\ooo$ in the domain of $\struct'$,  there is a unique interpretation $\struct$  with $\ooo^\struct=\ddd_o$ that  corresponds to $\struct'$.

\begin{defn}
We define the \emph{templified definition }$\D_{Temp}$ of $\D$ as the definition  $\{ \forall \ooo (\Psi^\ooo) \mid \Psi\in \D\}$ and we define $\Sigma_{Temp}'=\Sigma_{Temp}\cup \{P' \mid P\in \defp{\D}\}$.
\end{defn}
We assume that $\ooo$ consists only of first order predicate symbols. Under this condition, the templified definition $\D_{Temp}$  is a template over $\Sigma_{Temp}'$. 
\begin{proposition}\label{prop:OpensArgumentsRewrite}
Let $\struct$ be a well-founded model of $\D$ and $\struct'$ a well-founded model of $\D_{Temp}$ such that $\struct$ and $\struct'$ are identical on $\Sigma_{Temp}$. Then it holds that
\[P^{\struct} = \left\{ \ddd | (\ddd,\ooo^\struct) \in {P'}^{\struct'}\right\}\]
\end{proposition}
Stated differently, $\struct$ corresponds to $\struct'$. The templified definition captures the original one, and hence, each  theory can be rewritten in terms of the new templified defined symbols. 

\begin{proof}\label{proof:OpensArgumentsRewrite}
Assume that $\struct$ corresponds to $\struct'$.  It is easy to prove, by  induction on the formula structure, that for any formula $\varphi$ in the vocabulary of $\D$, it holds that $\varphi^\struct=(\varphi^\ooo)^{\struct'[\ooo:\ooo^\struct]}$.  We call this the independency property since it shows  that $(\varphi^\ooo)^{\struct'[\ooo:\ooo^\struct]}$ is influenced by only a small part of the interpretation of $P'$, namely  the values of domain atoms $P'(\ddd,\ooo^{\struct})$.

The key property to prove is that $\struct'$  is a  partial stable interpretation  of $\D_{Temp}$ iff for each value  $\ddd_o$ for $\ooo$, the unique  $\struct$ that corresponds to $\struct'$  such that $\ooo^\struct=\ddd_o$  is a partial stable interpretation of $\D$. Intuitively, a partial stable interpretation of $\D_{Temp}$ is a kind of union of partial stable interpretations of $\D$, one for each assignment of values to $\ooo$. 

We prove this property only in one direction. The other direction is similar.  Assume that $\struct'$ is a partial stable interpretation of $\D_{Temp}$ satisfying the three conditions of Definition~\ref{DefPartialStable}. We need to show for every $\struct$ that corresponds to $\struct'$, that $\struct$ is a partial stable interpretation of $\D$.  Condition 1), that $P(\ddd)^\struct=Max_\leq\{\varphi[\ddd]^\struct \mid \forall \xxx(P(\xxx)\rul\varphi[\xxx])\in \D\}$ follows from the fact that $P'(\ddd,\ooo^\struct)^{\struct'}$ satisfies the corresponding equation for $\D_{Temp}$,  that $P(\ddd)^\struct=P'(\ddd,\ooo^\struct)^{\struct'}$, and that for each rule body $\varphi$ for $P$, $\varphi[\ddd]^\struct=(\varphi^\ooo)^{\struct'[\ooo:\ooo^\struct]}$ (by the independency property). The condition 2) follows from the fact that when $T'=\{P'(\ddd,\ooo^\struct)\mid P(\ddd)\in T\}$, and $U'=\{P'(\ddd,\ooo^\struct)\mid P(\ddd)\in U\}$, then  $T$ is a  $\Tr$-set and $U$  a $\Un$-set of $\struct$ such that $\struct[T:\Un][U:\Tr]$ is closed under $\D$ iff  $T'$ is a  $\Tr$-set  and $U'$  a $\Un$-set of $\struct'$ such that $\struct'[T':\Un][U':\Tr]$ is closed under  $\D_{Temp}$.  This follows from the independency property. Condition 3) is proven similarly. 


It is easy to see that this property entails the proposition, since intuitively, it entails that a  well-founded model $\struct'$ of $\D_{Temp}$, which is the $glb_{\leqp}$ of all partial stable interpretations of $\D_{Temp}$ with the same context as $\struct'$, contains for each value $\ddd_o$ for $\ooo$ the $glb_{\leqp}$ of the partial stable interpretations $\struct$ of $\D$ in the context  with  $\ooo^\struct=\ddd_o$.
\end{proof}

\subsection{Simple Templates}
\ignore{In the Section \ref{sub:soidDef}, we built a logic \soids. }
Extending a logic with arbitrary (recursive) templates may easily
increase the descriptive  complexity of the logic.  Below, we develop
a simple but useful template formalism for \foid that does not have
this effect.  In addition, we show that libraries of simple templates
can be compiled away using them as macros.

In Figure~\ref{fig:BNFforL} we define sublanguages \foids, \esoids and
\asoids of \soids (by mutual recursion) consisting of atoms,
negations, conjunctions, quantification, definitions and the
let-construct. This last construct represents a second order
quantification, where the quantified symbol(s) $S$ are defined in an
accompanying paradox-free definition $\D$. Definitions of second order
symbols in \esoids and \asoids contain only (possibly nested) first
order definitions. Since model checking of (nested) first order
definitions is polynomial, the descriptive complexity of \foids is P,
of \esoids is NP and of \asoids is co-NP.

\begin{figure}[h]
\centering
\begin{subfigure}[t]{0.35\textwidth}
\begin{minipage}[t]{\linewidth}
\begin{align*}
\foids\ \varphi & ::= \\
& | s(\ttt) (\in Atom^1)\\
& | \lnot \varphi \\ 
& | \varphi \land \varphi \\
& | \exists_{\mathit{FO}}\ s : \varphi\\
& | \text{let }\{s(\ttt)\rul\varphi\}\text{ in }\varphi\\
& | \{s(\ttt)\rul\varphi\}\\
\end{align*}
\end{minipage}
\caption{\foids}\label{fig:FO}
\end{subfigure}
\begin{subfigure}[t]{0.3\textwidth}
\begin{minipage}[t]{\linewidth}
\begin{align*}
\esoids\ \epsilon & ::=\\
& | S(\ttt) (\in Atom^2)\\
& | \lnot \alpha \\ 
& | \epsilon\land\epsilon \\
& | \exists_{\mathit{FO}}\ s: \epsilon\\
& | \text{let }\{s(\ttt)\rul\varphi\}\text{ in }\epsilon\\
& | \{s(\ttt)\rul\varphi\}\\
& | \exists_{\mathit{SO}}\ s: \epsilon
\end{align*}
\end{minipage}
\caption{\soids}\label{fig:ESO}
\end{subfigure}
\begin{subfigure}[t]{0.3\textwidth}
\begin{minipage}[t]{\linewidth}
\begin{align*}
\asoids\ \alpha & ::= \\
& | S(\ttt) (\in Atom^2)\\
& | \lnot \epsilon \\ 
& | \alpha \land \alpha \\
& | \exists_{FO}\ s : \alpha\\
& | \text{let }\{s(\ttt)\rul\varphi\}\text{ in }\alpha\\
& | \{s(\ttt)\rul\varphi\}\\
& | \forall_{\mathit{SO}}\ s: \alpha
\end{align*}
\end{minipage}
\caption{\asoids}\label{fig:ASO}
\end{subfigure}
\caption{The \foids, \esoids and \asoids subformalisms of \soids\label{fig:BNFforL}}
\end{figure}

\begin{defn}
 A \textit{simple template} is a template of the form $\{\forall \xxx(P(\xxx)\rul\varphi_P[\xxx])\}$ with $P(\xxx)\in Atom^2$ and $\varphi_P\in \foids$.  
\end{defn} 
A simple template defines one symbol and contains one rule with an \foids body. Let $L$ be a template library over $\Sigma_{Temp}$ consisting of  non-recursive simple templates. Such a library is equivalent to the conjunction the completion of its definitions  $\forall \xxx(P(\xxx)\Leftrightarrow\varphi_P[\xxx])$. We want to show that while using such libraries  increases convenience, reuse, modularity, it does not increase complexity nor expressivity. Also,  such libraries can be used in the common way, as macros. 

\begin{thm} \label{thm:rewrite}For $\Sigma\cap\Sigma_{Temp}=\emptyset$, let $\varphi$ be a \esoids formula over $\Sigma\cup\Sigma_{Temp}$ that does not contain definitions of template symbols. There exists a polynomially larger \esoids formula $\varphi_1$ over $\Sigma$ that is $\Sigma$-equivalent to $\{\varphi\}\cup L$. There exists a polynomially larger \foids formula $\varphi_2$ over an extension $\Sigma_1$ of $\Sigma$ that is $\Sigma$-equivalent to $\{\varphi\}\cup L$.
\end{thm}
\begin{proof}
The formula $\varphi_1$ is obtained by treating $L$ as a set of macros. We iteratively substitute template atoms $P(\ttt)$ in $\varphi$ by $\varphi_P[\ttt]$. This process is equivalence preserving. It  terminates due to the stratification condition on $L$, and the limit is a polynomially larger \esoids formula $\varphi_1$ in the size of $\varphi$ (exponential in  $\#(L)$) that is $\Sigma$-equivalent to  $\{\varphi\}\cup L$. 

To obtain $\varphi_2$, we apply the well-known transformation of moving existential quantifiers to the front and skolemising them. Second order quantifiers can be switched with  first order ones using:
\begin{gather*}
\forall_{FO}\, x:\exists_{SO} P : \varphi \Leftrightarrow \exists_{SO} P' : \forall_{FO}\, x : \varphi[P(\ttt) \text{\textbackslash} P'(\ttt,x)  ]
\end{gather*}
This process preserves $\Sigma$-equivalence. As only a polynomial number of steps are needed to transform the formula into this desired state, the size of the resulting formula is  polynomially larger.
\end{proof}

Previous results in \cite{nmr/IanniIPSC04} indicated that the introduction of simple, stratified templates does not introduce a significant performance hit. The above theorem recovers these efficiency results. 

\section{Conclusion}
In this paper we developed a new way to define language constructs for a logic.
New language constructs must combine a syntactical rule with a three-valued semantic evaluation. This three-valued semantic evaluation is subject to certain restrictions.
Language constructs can then be arbitrarily combined to compose a logic.
In particular, we construct $\soids$: a second order language with inductive definitions. 

Using this language, it is easy to define templates as second order definitions.
We conclude our paper with a rewriting scheme to show that, given some restrictions, templates do not increase the descriptive complexity of the host language. 

In the future, we want to generalize our way of defining language constructs to allow functions and provide a more comprehensive type system.
On the more practical side, we intend to bring our ideas into practice by extending the IDP\cite{url:idp} system with simple templates.

\clearpage
 \bibliographystyle{acmtrans}
\bibliography{idp-latex/krrlib}
\end{document}